\documentclass[a4paper]{scrartcl}
\usepackage[utf8]{inputenc}

\usepackage{pdfpages}
\usepackage{amsmath,amsfonts,amssymb,amsthm}
\usepackage{tikz}
\usetikzlibrary{arrows,decorations.pathmorphing}
\usepackage{ulem}
\usepackage{subcaption}
\usepackage[algo2e,vlined,linesnumbered,ruled]{algorithm2e}
\tikzstyle{vertex}=[circle, minimum size=15pt,inner sep=0pt]
\tikzstyle{smallvertex}=[circle,minimum size=1pt,inner sep=0pt]
\tikzstyle{blackedge} = [draw,thick,-stealth']

\tikzstyle{priority1} = [draw,thick,-stealth', draw=black]
\tikzstyle{priority2} = [draw,thick,-stealth', draw=red, decoration={snake, amplitude=2pt, post length=4pt,
    pre length = 0.1pt}, decorate]
\tikzstyle{priority3} = [draw,thick,-stealth', draw=yellow, dotted]
\tikzstyle{priority4} = [draw,thick,-stealth', draw=green, ,decoration={snake, post length=3pt,
    pre length = 0.1pt}]
\tikzstyle{priority5} = [draw,thick,-stealth', draw=blue, ,decoration={coil,segment length=3.5pt, post length=3pt,
    pre length = 0.1pt}]

\newcommand{\N}{\ensuremath{\mathbb{N}}}
\renewcommand{\O}{\ensuremath{\mathcal{O}}}

\newtheorem{theorem}{Theorem}
\newtheorem{lemma}[theorem]{Lemma}
\newtheorem{conjecture}[theorem]{Conjecture}
\newtheorem{corollary}[theorem]{Corollary}

\newtheorem{observation}[theorem]{Observation}
\newtheorem{proposition}[theorem]{Proposition}

\definecolor{Martin}{rgb}{0,0.6,0.3}
\definecolor{Laura}{rgb}{0.5,0,0.5}
\definecolor{Robert}{rgb}{0,0,1}

\newcommand*{\Titel}{Nash equilibria in routing games with edge priorities}
\title{\Titel}
\author{Robert Scheffler\thanks{BTU Cottbus-Senftenberg, \texttt{robert.scheffler@b-tu.de}} \and Martin Strehler\thanks{BTU Cottbus-Senftenberg, \texttt{martin.strehler@b-tu.de}} \and Laura Vargas Koch\thanks{RWTH Aachen, \texttt{laura.vargas@oms.rwth-aachen.de}} }
\date{}

\usepackage{scrlayer-scrpage}
\rehead{\Titel}
\lohead{R. Scheffler, M. Strehler, and L. Vargas Koch}
\begin{document}
\maketitle

\begin{abstract}
In this paper we present a new competitive packet routing model with edge priorities. We consider players that route selfishly through a network over time and try to reach their destinations as fast as possible. If the number of players who want to enter an edge at the same time exceeds the inflow capacity of this edge, edge priorities with respect to the preceding edge solve these conflicts. Our edge priorities are well-motivated by applications in traffic. For this class of games, we show the existence of equilibrium solutions for single-source-single-sink games and we analyze structural properties of these solutions. We present an algorithm that computes Nash equilibria and we prove bounds both on the Price of Stability and on the Price of Anarchy. Moreover, we introduce the new concept of a Price of Mistrust. Finally, we also study the relations to earliest arrival flows.

\end{abstract}

\section{Motivation}
\label{sec:motivation}
Routing games over time are widely studied \cite{cominetti2011existence, hoefer2009competitive, KS09} due to various applications, e.g., road and air traffic control, logistic in production systems, communication networks like the internet, and financial flows. Depending on applications, there is a huge range of models. There are games with non-atomic \cite{KS09} as well as atomic players \cite{hoefer2009competitive}, further there are games with continuous time \cite{hoefer2009competitive} as well as games with discrete time steps \cite{HPSV16}. 

Usually, players are allowed to enter an edge directly at their arrival time. An essential ingredient of these models is the priority rule for tie-breaking. If two or more players arrive at a node at the same time, one needs to decide which player may enter the subsequent link first. Some models maintain the overall arrival order with a first-in first-out rule \cite{cominetti2011existence, hoefer2009competitive, KS09}, i.e., only players with the exact same arrival time are re-arranged. Other models allow overtaking by higher prioritized players with global priority lists \cite{HPSV16,hoefer2009competitive} or local priority lists \cite{HPSV16} of players, that is, priority is a player inherent property and players are linearly ordered.

Yet, a very natural way to motivate priority has origin in traffic flows. At intersections of road networks, priority is usually determined by road signs. That is, there is one road with the way of right, whereas the traffic on the crossing road is prompted to give way by yield or stop signs.
This motivates a priority that depends on the chosen edges and not only on the players.

In this paper, we introduce a routing game over time on a network $G=(V,E)$ where priority on every edge is assigned to players depending on their preceding edge. Every edge $e=(v,w)$ has an ordering of the $v$ entering edges. That is, if two players arrive at a node $v\in V$ at the same time via two different edges $e_1\in E$ and $e_2\in 
E$, and both players are going to enter the same outgoing edge $e$, then the player from the edge that is first in the order of edge $e$, say $e_1$, always moves first and the 
player on the other edge, here $e_2$, always has to wait.  Moreover, if there arrives a third player in the subsequent time step at $v$ via edge $e_1$, then this player also gets 
priority and the player on edge $e_2$ has to wait another round. In other words, for this edge-based priority the first-in first-out (FIFO) property does not hold in nodes. 
Nevertheless, we assume that the first-in-first-out property holds on edges, i.e., players may queue up before traversing an intersection. 

\paragraph*{Our Contribution} In this paper, we focus on the symmetric game, i.e., there is only one source and one sink for all the players. 
We show that Nash equilibria are guaranteed to exist in these symmetric games and we present a tight bound on the Price of Anarchy as well as a significant gap between the best Nash equilibrium and the social optimal state. 
 Surprisingly, even in the best Nash equilibrium, player strategies may include the use of cycles. We even show that some players may have to visit certain nodes in the network up to ${\cal O}(k)$-times in a network with $k$ players and $m\in {\cal O}(k)$ edges.

Furthermore, we introduce the new concept of a Price of Mistrust. Usually, the Price of Anarchy (PoA), i.e., the ratio between worst equilibrium and system optimum, is obtained in somewhat strange constellations, e.g., two players are blocking each other. On the other hand, achieving the Price of Stability (PoS), i.e., the ratio between best equilibrium and system optimum, seems to require at least some co-operation or friendly players. Since we consider a routing game over time, players start one after another, and we assume that each player fixes its strategy just when actually starting in the source node. We assume all strategies of players that started earlier to be fixed and known and each player can react on the choices of the preceding players. Hence, only the strategies of subsequent players are unknown. Here, a player mistrusts the following players and chooses the best response that can be delayed by as few of the subsequent players as possible. We show that this concept yields equilibria with values proper between PoS and PoA. Furthermore, we present an algorithm to compute mistrustful equilibria for symmetric routing games.

The paper is organized as follows. In Section~\ref{sec:lit}, we discuss related results from literature and categorize our problem. Section~\ref{sec:model} gives a formal problem definition and fixes the terminology. Afterwards, we present basic results on Nash Equilibria in our model in Section~\ref{sec:results}. In Section~\ref{sec:alg} we present an algorithm that computes such equilibria and introduce the Price of Mistrust. In Section~\ref{sec:eaf} we show the connection to earliest arrival flows. Finally, we give a short outlook on some extensions and close the paper with a discussion in Section~\ref{sec:disc}.

\section{Pointers to the Literature}
\label{sec:lit}
Routing games are a widely used approach to model traffic situations and there are a lot of different variants. Rosenthal~\cite{rosenthal1973network} introduced atomic selfish routing games in 1973. Several work on the existence and efficiency of equilibria in this kind of games has been done~\cite{anshelevich2008price, awerbuch2005price, christodoulou2005price, roughgarden2002bad}. For a profound survey by Roughgarden see Chapter~18 in~\cite{nisan2007algorithmic}. 

However, these routing games have the drawback to be static. Hoefer et al.~\cite{hoefer2009competitive} introduced temporal network congestion games, where players block each other only when they enter or use an edge at the same moment in time. They analyze the existence of equilibria as well as best response dynamics for several priority rules. Harks et al.~\cite{HPSV16} proved bounds on the price of anarchy and stability for competitive packet routing games where players have a global or local ranking. In contrast to the player dependent forwarding/priority rules in the work just mentioned, Cao et al.~\cite{cao2017network} use  FIFO policy with an edge dependent tie-breaking. They focus on the comparison between Nash equilibria and equilibria in a game with full control (subgame perfect equilibria) on acyclic graphs. In the paper at hand, we will use the subordinated tie-breaking of Cao et al.~\cite{cao2017network} as our main criterion in a non FIFO setting similar to the model in Harks et al.~\cite{HPSV16}.  

While Roughgarden's model is a game theoretical but static model, Ford and Fulkerson introduced a dynamic model for \emph{Flows over time}, also called \emph{dynamic network flows}, in their seminal works~\cite{FF58,FF62}. The authors used time-expanded networks to keep track of the dynamic movement of flow particles and compute system optimal solutions.  A very fascinating concept in this context is the \emph{earliest arrival flow} (EAF). Given a network, an earliest arrival flow simultaneously maximizes the amount of flow that has already reached the sink for all points in time. While Gale~\cite{gale59} has shown that earliest arrival flows always exist in single source single sink networks and it is even possible to derive a fully-polynomial time approximation scheme in the multiple source single sink case~\cite{FS07}, earliest arrival flows do not necessarily exist in a network with multiple sinks~\cite{BaumannSkutella}. An excellent introduction to dynamic flows is given by Skutella~\cite{Skutella2009}.
A game theoretical extension of this model was given by Koch and Skutella~\cite{KS09,KS11}, who presented a characterization of Nash equilibria for flows over time. However, the results were obtained for non-atomic players, i.e., flow is arbitrarily splittable, while we consider atomic players in our work.

Independently, similar settings were investigated from a more application-driven perspective. In the traffic community, Yagar~\cite{yagar} introduced \emph{dynamic traffic assignment} (DTA) already in the early 1970s. Nowadays, DTA includes a wide variety of problems and concepts~\cite{dtaprimer}, e.g., also departure times are subject to optimization. In this context, \emph{dynamic network loading}~\cite{WuChenFlorian} addresses the non-trivial task of computing a single dynamic assignment based on load-dependent travel time functions. Yet, due to the high complexity of urban traffic networks with traffic signals et cetera, most approaches rely on heuristics or simulation-like techniques or they restrict the problem a-priori, e.g., by drastically limiting the set of feasible paths. 

\section{A Model for Routing with Edge Priorities}
\label{sec:model}

The object under consideration in this paper is a \emph{routing game with edge priorities} which we define as follows. 

\subsection{Playing field and rules}

The routing game is played on a directed network $G=(V, E)$, where $V$ is the node set with $n=|V|$ nodes and $E$ is the edge multiset with $m=|E|$ arcs. We allow multiple edges, i.e., more than one edge starting at the same node and ending at the same node, as well as loops, i.e., an edge starting and ending at the same node.
Further, $u:E\to \N$ are \emph{integral capacities} on the edges. This capacity limits the inflow rate, i.e., the amount of flow entering an edge $e\in E$ per time unit. Furthermore, the edges of $G$ are equipped with \emph{integral transit times} $\tau :E\to\N_0$. The transit time or costs $\tau(e)$ denotes the time a player needs to traverse edge $e\in E$. We use constant transit times here, that is, players use edges independently of other players and there is no delay due to congestion.
We allow edges with transit time zero, since they are useful to model capacities on nodes. The throughput capacity of a node $v$ can be limited by replacing this node by an edge $(v',v'')$ with the desired capacity and costs zero. Similar to the work of Harks et al.~\cite{HPSV16}, we restrict to integral transit times, since each player blocks exactly one capacity unit for one time unit on each edge of its path. 

Additionally, each routing game fixes two distinguished nodes, a source $s$ and a sink $t$. We assume w.l.o.g.\ that the source has no incoming edges, i.e., $\operatorname{deg^-}(s)=0$. Every \emph{player} in the set $N=\{1, \dots , k\}$ of players chooses a path $P_i$ from the set of $s$-$t$-paths $\mathcal{P}_{st}$ and travels over the time through the network. To be more specific, a player can leave an edge $e$ at the earliest $\tau(e)$ time units after entering $e$. Moreover, we only consider discrete time steps, since we have integral transit times. Please note, that we do not restrict to simple paths, since it may be advantageous for a player to choose a path containing a cycle as we will see in the upcoming analysis.
Since all the players have the same set of strategies we call it a \emph{symmetric} or \emph{single-commodity} game. 

When traveling through the network over time, it might happen that more than $u(e)$ players try to enter an edge $e=(v,w)$ at the same time. To decide which players are allowed to proceed directly and which players need to wait at least one time unit, we define a \emph{priority order} $\pi(e)=(e_1,e_2,\dots,e_{\operatorname{deg^-}(v)})$, i.e., for every $v\in V$ we assign an ordered list $\pi(e)$ of all incoming edges $e_i=(u_i,v)\in\delta^-(v)$ of $v$ to each outgoing edge $(v,w)\in\delta^+(v)$ of $v$. If edge $e$ has remaining capacity at time $T$, a player seeking to enter edge $e$ at time $T$ may do so, if the incoming edge of this player has the lowest possible index in the ordered list $\pi(e)$ among all players who want to enter the link $e$. This applies iteratively. Thus, after the first player has entered the link, we choose the next player with the lowest possible incoming edge in $\pi(e)$ from the remaining players, if $e$ still has capacity left.

Among the players waiting on an edge $e'$ the \emph{first-in first-out} rule (FIFO) applies. This means, if Player~$i$ and Player~$j$ both try to enter edge $e$ from the same edge $e'$, the player who arrived on edge $e'$ first will be preferred. If several players have entered $e'$ at the same time, but the desired edge $e$ does not provide enough capacity for all of them, we use the number of the player as a global tie-breaker. That is, in case of a tie, Player~$i$ moves before Player~$j$ if $i<j$. Especially, this rule applies at the source, i.e., Player~1 is always the first player to leave source $s$.

If for every node $v \in V$ it holds $\pi (e)=\pi (e')$  for all priority lists of outgoing edges $e, e' \in \delta^+(v)$, then we call a game \emph{global} and else \emph{local}. In case of a global game, we may simply define a total order on all edges, since such an order canonically defines the priority list of each edge. This total order is of course not unique. For the sake of simplicity in a global game, we always relabel the edges of $G$ to $e_1,e_2,\dots,e_m$  such that $e_i$ has higher priority than $e_j$ whenever $i<j$.

Summarizing, we determine priority in the order \emph{edge list} $>$ \emph{FIFO} $>$ \emph{player ID}. All together, a \emph{routing game with edge priorities} is defined as a tuple $(G, u,  \tau, N, \pi)$.

\subsection{Goal of the game}

In the game, a \emph{strategy} of a Player $i$ is an $s$-$t$ path $P_i$. Let $P$ be the \emph{profile} or \emph{state} of the game with the strategies of all the players, that is, $P$ consists of $k$ paths $P_1$ to $P_k$. 
Now, we denote the \emph{arrival time} of Player~$i$ as $C_i(P)$, which consists of the transit times $\tau(e)$ of all the edges of the chosen path $P_i$ and the \emph{waiting time} on those edges.
Obviously, the former is independent of the strategies of the other players due to constant transit times, but the latter significantly depends on $P$. 

A profile $P$ is \emph{socially optimal} if it minimizes the total costs $C(P)=\sum_{i \in N} C_i(P)$.
However, we assume players to behave \emph{selfishly}, i.e., each player aims to minimize its own arrival time. We call a state a (pure) Nash-equilibrium (PNE) if the chosen strategies separately minimize the players' costs for every player. Let $P_{-i}$ be the state $P$ without the strategy of Player~$i$. Furthermore, with $P_i',P_{-i}$ we denote replacing the strategy of Player~$i$ in $P$ by $P_i'$.  More formally, a routing game with state $P$ and strategy $P_i$ for Player~$i$ is in a PNE if
$ C_i(P_i, P_{-i}) \leq C_i(P_i', P_{-i}) \quad \forall P_i' \in \mathcal{P}_{s t}$, for all players $i \in N$.
In other words, no player can reduce its costs by switching from $P_i$ to another path~$P_i'$.

\subsection{Making the game well-defined}\label{sec:zerocycles}

For a well-defined game, it is important to have a unique mapping of a strategy profile $P$ to costs $C(P)$ of the players. Unfortunately, the current model can still lead to some paradoxical situations in connection with zero costs cycles. 
An example with two players is given in Figure~\ref{fig:zerocycle}. Assume, Player~1 chooses the path $(s,v_1,v_3,v_2,v_4,t)$ and Player~2 chooses $(s,v_2,v_4,v_1,v_3,t)$. Note that the paths intersect twice. 
Player~$i$ hits node $v_i$ before node $v_{3-i}$ on the respective path. 

Now, assume the red wavy edges $(v_4,v_1)$ and $(v_3,v_2)$ in the cycle have priority over the black straight entering edges $(s_1,v_1)$ and $(s_2,v_2)$ to proceed on $(v_1,v_3)$ and $(v_2,v_4)$, respectively. Furthermore, all edges have zero transit time. On the one hand, Player~1 could reach $v_2$ in zero time and block Player~2 there. If Player~2 is blocked at $v_2$, it cannot block Player~1 at $v_1$. Thus, Player~1 reaches $t_1$ at time 0, and Player~2 reaches $t_2$ at time 1. On the other hand,  the same argument is valid vice versa. Since Player~2 can reach $v_1$ in zero time, Player~1 is blocked. Hence, Player~2 reaches $t_2$ at time 0, and Player~1 reaches $t_1$ at time 1. In other words, there is no unique embedding of the paths and therefore there is no unique mapping from the strategy $P$ to arrival times of players $C(P)$.

\begin{figure}[ht!]
\centering
\begin{tikzpicture}[every node/.style={circle, minimum size=15pt, inner sep=2pt, font = \footnotesize}]
\node[draw] (s1) at (0,0) {$s$};
\node[draw] (t2) at (0:4) {$t$ };
\node[draw] (v1) at (0:1) {$v_1$};
\node[draw] (v3) at (90:1) {$v_3$};
\node[draw] (v2) at (180:1) {$v_2$};
\node[draw] (v4) at (270:1) {$v_4$};

\draw[priority1] (s1) to (v1);
\draw[priority1] (v4) to (t2);
\draw[priority1] (v3) to (t2);
\draw[priority1] (s1) to (v2);

\draw[priority1] (v1) to (v3);
\draw[priority2] (v3) to (v2);
\draw[priority1] (v2) to (v4);
\draw[priority2] (v4) to (v1);
\end{tikzpicture}
\caption{Example for embedding problems with zero costs cycles. Red wavy edges have higher priority.}\label{fig:zerocycle}
\end{figure}
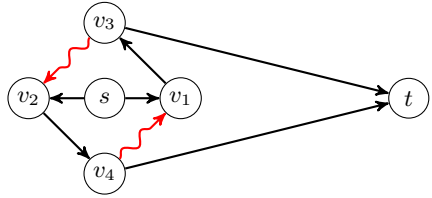

Since we do not want to forbid zero transit times in general, we exclude all networks with directed zero costs cycles from our consideration. Moreover, we will compute ratios of various solutions, e.g., the Price of Anarchy. To avoid division by zero, we additionally exclude all games where source $s$ and sink $t$ have distance zero.

\section{Basic Results on Equilibria}
\label{sec:results}
In this section, we highlight some major properties of equilibria in routing games with edge priorities. We show that a PNE always exists in every symmetric game and we bound its costs with respect to the social optimum.

\subsection{Existence and Uniqueness}

First, we claim existence of a pure-strategy Nash equilibrium (PNE) in a symmetric game. However, a proof will be given by Theorem~\ref{theo:algo:pathfinder} in Section~\ref{sec:alg}, where we prove the correctness of an algorithm computing equilibria.

\begin{theorem}
\label{theo:existence}
In every symmetric routing game with edge priorities there exists a PNE.
\end{theorem}

An equilibrium in terms of the profile is not unique. But for a given game, also the value $C(P)$ of equilibria is not unique.  

\begin{proposition}\label{prob:notunique}
The value of equilibria in a routing game with edge priorities is not unique.
\end{proposition}

\begin{proof}
We use a graph based on the famous example of Braess~\cite{Braess}. We refer to the graph with $b$ parallel paths from $s$ to $t$ as \emph{$b$-Braess graph}. Each path consists of three edges. Furthermore, edges  connect the third node of the $i$-th path with the second node of the $(i+1)$-th path for all $1\le i<b$. A $4$-Braess graph is depicted in Figure \ref{fig:pnenotunique}. Consider the $b$-Braess graph where all $s$-leaving edges have transit time one and all the other edges have costs zero, while all edges have unit capacity. In this network, priority follows the scheme depicted in Figure~\ref{fig:pnenotunique}, that is, the red wavy edges connecting the parallel paths are always prioritized over the edges in the direct paths.

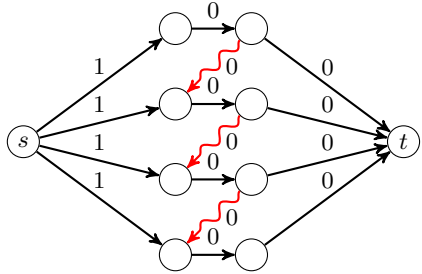
\begin{figure}[ht!]
\centering
\begin{tikzpicture}[every node/.style={circle, minimum size=12pt, inner sep=2pt, font = \footnotesize}]
\node[draw] (s) at (0,1.5) {$s$};
\node[draw] (v0) at (2,3) {$ $};
\node[draw] (v1) at (3,3) {$ $};
\node[draw] (v2) at (2,2) {$ $};
\node[draw] (v3) at (3,2) {$ $};
\node[draw] (v4) at (2,1) {$ $};
\node[draw] (v5) at (3,1) {$ $};
\node[draw] (v6) at (2,0) {$ $};
\node[draw] (v7) at (3,0) {$ $};
\node[draw] (t) at (5,1.5) {$t$};

\draw[priority1] (s) to node[above] {1} (v0); 
\draw[priority1] (s) to node[above] {1} (v2); 
\draw[priority1] (s) to node[above] {1} (v4); 
\draw[priority1] (s) to node[above] {1} (v6); 

\draw[priority2] (v1) to node[right] {0} (v2); 
\draw[priority2] (v3) to node[right] {0} (v4); 
\draw[priority2] (v5) to node[right] {0} (v6); 

\draw[priority1] (v0) to node[above] {0} (v1);
\draw[priority1] (v2) to node[above] {0} (v3); 
\draw[priority1] (v4) to node[above] {0} (v5); 
\draw[priority1] (v6) to node[above] {0} (v7); 

\draw[priority1] (v1) to node[above] {0} (t); 
\draw[priority1] (v3) to node[above] {0} (t); 
\draw[priority1] (v5) to node[above] {0} (t); 
\draw[priority1] (v7) to node[above] {0} (t); 
\end{tikzpicture}
\caption{The 4-Braess graph, where the value of a PNE is not unique. Red wavy edges have higher priority.}\label{fig:pnenotunique}
\end{figure}

We study the game with $k=b$ players. Obviously, using the $b$ parallel paths in this network is a PNE. No player can improve its arrival time by switching to another path. 
This PNE is also socially optimal, yielding total costs $C(P)=k$.

However, it is also a PNE for all players to go along the zigzag-path, i.e., every player uses all the red wavy edges, in order of the player IDs. If Player~1 chooses this path, its arrival time is still~1. Player~2 cannot arrive at time~1 with this choice of Player~1, so Player~2 can decide to follow Player~1 arriving at time~2. The same argument applies to all other players. In this case, Player~$i$ has costs $i$ and no improving move. In total $C(P)=\frac{k(k+1)}{2}$.
\end{proof}

This result also holds if we forbid edges with costs equal to zero. Change the transit times in the $b$-Braess graph as follows. The outgoing edges of $s$ get the costs $1, 3, \ldots, 2b - 1$ from top to bottom. Contrary, the ingoing edges of $t$ get these costs from bottom to top. Each other edge is assigned transit time 1. Now, the parallel paths have the same costs like the zigzag-path $c(P) = 2b + 1$. Thus, the same argument as above applies.

\subsection{Structural properties}

The previous result was to be expected, since the costs of a PNE are also not unique in many other routing games with different priority settings or in equilibrium flows with capacities. More surprising, in routing games with edge priorities, a player may visit a node arbitrary often. Furthermore, there are networks where this can happen in every PNE. Before we come to this statement, we start with the following lemma.

\begin{lemma}\label{lemma:zeitpunkte}
In a PNE in a symmetric routing game with edge priorities, the starting order of the players is maintained at the sink $t$. Further, for every $1 \leq i \leq k-1$ Player $i+1$ arrives at $t$ at most one time unit after the preceding Player $i$. 
\end{lemma}

\begin{proof}
Assume for the sake of contradiction, Player $i$ starts before Player $j$, i.e., $i<j$, but Player $i$ reaches the sink $t$ after Player $j$. Then $i$ can copy the strategy of $j$. Since Player $i$ is allowed to start before Player $j$, $i$ will always be before $j$, improving the arrival time in $t$. Hence, the initial state was not a PNE.

Furthermore, assume that there is a time gap of at least two time units between two arriving players. Choose the smallest index, such that Player $i+1$ arrives more than one time unit later than Player $i$. Note that all Players $j$ with $j\le i$ also arrive at $t$ at time $T$ at the latest.  
In the following, we construct a path $P_{i+1}'$ for Player $i+1$, which realizes an arrival time of at least $T+1$ where $T$ is the arrival time of Player $i$. The basic idea is to follow the preceding players, i.e., to use the same edges as the tracked player exactly one time step later. 

Player $i+1$ starts with following Player $i$. If it can follow this player until it reaches the sink, Player $i+1$ arrive at time $T+1$. Thus, assume it cannot follow Player $i$ since Player $i+1$ is blocked at some node $v$ by Player $j$, who has the right to enter the next edge $e=(v,w)$ first. If $j>i+1$, Player $I+1$ can copy the the strategy of $j$ from $s$ to $v$ and prevent the blocking. From node $v$ on it continues following Player $i$. If $j<i$, it simply follows $j$ from now on. 

We iterate this idea, i.e., whenever Player $i+1$ is blocked, Player $i+1$ either copies the first part of the strategy of the blocking Player $j$ for $j>i+1$ or it changes the target whom it follows to $j$ for $j\le i$. Even when a node is visited twice, time will have advanced at least one time unit, since there are no zero costs cycles and hence, a finite path is constructed. Therefore, Player $i+1$ reaches sink $t$ by following a Player $j$ with $j<i+1$.  Consequently, Player $i+1$ reaches $t$ at time $T+1$ at the latest.   
\end{proof}

Note that this statement does not hold for intermediate nodes. There, players may arrive much earlier on a subordinate edges and wait until they can proceed. 

\begin{theorem}\label{prob:multinodes}
For every $k \in \N$, $k \geq 3$, there exists a routing game with edge priorities with $k$ players on a network with $2k-1$ edges and unit capacities $u\equiv 1$ such that there is a player who visits a node $k-1$ times in every PNE. 
\end{theorem}

\begin{proof}
We construct an instance with a unique PNE. The network consists of three nodes $s$, $v$, and $t$. The source $s$ is connected to $v$ with $k$ parallel edges $e_{k},\dots,e_{2k-1}$. Furthermore, there are $k-2$ loops at $v$, namely $e_2,\dots,e_{k-1}$, and a single edge $e_1$ from $v$ to the sink $t$. All edges have capacity $u\equiv 1$ and transit time $\tau\equiv 1$. An example of the network is given in Figure~\ref{fig:kreis}. We use global edge priorities, i.e., the index of the edge is equivalent to its priority. 

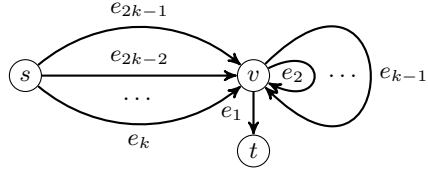
\begin{figure}[ht!]
\centering
\begin{tikzpicture}[every node/.style={circle, minimum size=12pt, inner sep=2pt, font = \footnotesize}]
\clip(4.5,4.75) rectangle (11,7);
\node[draw] (s) at (5,6) {$s$};
\node[draw] (m) at (8,6) {$v$};
\node[draw] (t) at (8,5) {$t$};

\draw[bend angle=40, bend left, blackedge] (s) to node[above, inner sep=0pt, rectangle] {$e_{2k-1}$} (m); 
\draw[blackedge] (s) to node[above, inner sep=0pt, rectangle] {$e_{2k-2}$} (m); 
\draw[bend angle=40, bend right, blackedge] (s) to node[below, inner sep=0pt, rectangle] {$e_{k}$} (m); 
\node(dots) at (6.5,5.7) {$\dots$};

\draw[blackedge] (m) .. controls (9,6.5) and (9,5.5) .. node[left] {$e_2$} (m);
\draw[blackedge] (m) .. controls (10,8.0) and (10,4.0) .. node[right] {$e_{k-1}$} (m);
\node(dots) at (9.20,6) {$\dots$};

\draw[blackedge] (m) to node[left] {$e_1$} (t); 
\end{tikzpicture}
\caption{This example shows that there is a PNE such that a player visits a node multiple times. All edges have unit travel time. We use a global priority list where $e_i$ is preferred over $e_j$ if $i<j$.}\label{fig:kreis}
\end{figure}

First, we define a state $P$ in which Player $k-1$ visits node $v$ exactly $k-1$ times and prove that this is indeed a PNE. Afterwards, we show that the strategies of the Players 1 to $k-1$ coincide in every PNE.
For Player $i \in \{ 1 , \ldots ,k-1 \}$ define $P_i=(e_{k+i-1}, e_i , \ldots , e_1)$. Player $k$ can choose any path arriving at time $k+1$. By construction, Player $i$ visits node $v$ exactly $i$ times and arrives at time $i+1$ for $i \in \{1, \ldots , k-1 \}$. Assume for the sake of contradiction that the state $P=(P_1, \ldots , P_k)$ is not a PNE. Then there is a player that can improve her strategy. Since the players arrive one after another, improving is only possible by delaying a preceding player. This is not possible, since every player goes along edges with the highest possible priorities. 

Assume there is a PNE $P'$ which differs from $P$ by another path than $P_k$. Let $i$ be the first player deviating from the strategy in $P$. Player $i$ cannot arrive earlier than in $P$ as argued above. Moreover, Player $i$ cannot arrive later than in profile $P$ in a PNE, since strategy $P_i$ is always a valid choice for arriving at time $i+1$. Thus, $P'$ realizes the same arrival time as $P_i$. In this case, Player $i+1$ can now improve by playing $(e_x, e_i , \ldots , e_1)$ where $e_x$ is the unused edge among $e_{k+i-1}$ and  $e_{k+i}$ with the lower index. With this strategy, Player $i+1$ overtakes Player $i$ either before entering the loops (if Player $i$ does not choose $e_{k+i-1}$) or while traversing the loops. Thus, there is an improving move for Player $i+1$ and $P'$ was not a PNE.
\end{proof}

Note that we used multi-edges for simplicity, here. In a network without multi-edges, one can still achieve ${\cal O}(k)$ visits of a node by a single player by subdividing each edge. This result is also not limited to networks with unit capacities. By using a single edge $(s,v)$ with higher capacity in the example in Figure~\ref{fig:kreis}, one can achieve a similar result. 

Contrariwise, there is no such result on the edges, like we show in Corollary~\ref{cor:double_edges} in Section~\ref{sec:alg}. Since edges and FIFO determine priority, there is no advantage of using an edge twice.

\subsection{Bounding the PoS and the PoA}

On the one hand, we have seen in Lemma~\ref{prob:notunique} that there may exist several user equilibria with different values for a routing game with edge priorities. Consequently, one may wonder how bad this selfish routing can be. What is the worst value of an equilibrium compared to the social optimum? This ratio is known as the 
\emph{Price of Anarchy} (PoA) [cf. \cite{koutsoupias1999worst}]. On the other hand, if a social optimal state is not an user equilibrium, then it is not stable, i.e., some players have occasion to change route. In many cases this will worsen the value of the solution. If every social optimal solution is not stable, then even the best equilibrium has additional costs compared to the social optimum. This ratio is called the \emph{Price of Stability} (PoS) [cf. \cite{anshelevich2008price}].   

In this subsection, we give a tight bound of the PoA for routing games with edge priorities. Furthermore, we show that there are instances of a routing game with edge priorities where every social optimum is not stable and we present a lower bound example of the PoS.  

\begin{theorem}\label{theo:poa}
The Price of Anarchy in a symmetric routing game with edge priorities is at most $\frac{k+1}{2}$, where $k$ is the number of players. This upper bound is tight.
\end{theorem}

\begin{proof}
The proof uses a very similar argument as used in \cite{HPSV16} to show a result for the PoA in routing games with player priority lists. In the first part of the proof, we show the upper bound on the price of anarchy $PoA \leq \tfrac{k+1}{2}$.

Let $\ell$ be the length of a shortest $s$-$t$ path. We show by induction that the costs of player $i$, that is, the $i$-th player in the global tie-breaking, in a PNE is bounded by $\ell+(i-1)$.

First note that in any PNE the first Player 1 in the tie-breaking order will always choose a shortest path and arrive at time $\ell$. If this player arrives at a later moment in time, it is an improving move to switch to the shortest path with highest possible priority. For the existence of such a path, we again refer to Algorithm~\ref{algo:pathfinder}.

For player $i$ assume that the claim holds true for the Players $i-1$, i.e., Player $i-1$ arrives no later than $\ell+(i-2)$. Lemma~\ref{lemma:zeitpunkte} implies that Player $i$ arrives at most one time unit after Player $i-1$ in every PNE. Hence, Player $i$ arrives at time $ \ell+(i-1)$ at the latest. 

This yields an upper bound on the costs of an equilibrium $P$, namely 

\[C(P)\le \sum_{i \in N} \left(\ell + (i-1)\right) = k\ell \cdot \frac{k(k-1)}{2}.\]

In an optimal solution, each of the $k$ players needs to traverse at least a shortest path of length $\ell$. Therefore, the costs of a social optimal solution are at least $k\ell$.

For the price of anarchy it follows that
$$
\operatorname{PoA} \leq \frac{k\ell+\frac{k(k-1)}{2}}{k\ell}=1+\frac{k-1}{2\ell},
$$
which is maximal for $\ell=1$ and provides an upper bound of $\frac{k+1}{2}$ on the Price of Anarchy in any routing game with edge priorities.

Moreover, we have already shown an example in the proof of Lemma~\ref{prob:notunique}, where the price of anarchy is equal to $\frac{k+1}{2}$. Hence, this bound is tight.
\end{proof}

For similar routing games, e.g.~\cite{HPSV16}, it turns out that there is a stable socially optimal routing, i.e., the PoS is equal to 1. Surprisingly, this is not true for the routing game with edge priorities and the PoS can be of the same order of magnitude as the PoA.

\begin{theorem}\label{theo:pos}
The Price of Stability in a symmetric routing game with edge priorities can be $\geq \frac{k+1}{4}$, where $k$ is the number of players.
\end{theorem}

\begin{proof}
We consider a $b$-Braess graph with $b$ parallel paths and unit capacities. All outgoing edges of $s$ and all but the lowest incoming edge of $t$ have travel time 1. All other edges, in particular the third edge of the $b$-th path, have travel time 0. The 4-Braess graph with these travel times is shown in Figure~\ref{fig:pos}. The red wavy edges connecting the parallel paths have the higher priority.

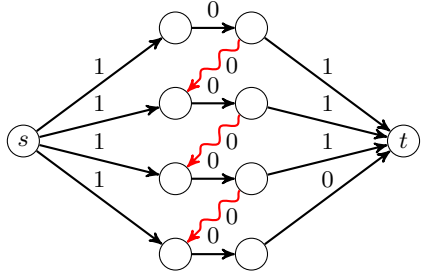
\begin{figure}[ht!]
\centering
\begin{tikzpicture}[every node/.style={circle, minimum size=12pt, inner sep=2pt, font = \footnotesize}]
\node[draw] (s) at (0,1.5) {$s$};
\node[draw] (v0) at (2,3) {$ $};
\node[draw] (v1) at (3,3) {$ $};
\node[draw] (v2) at (2,2) {$ $};
\node[draw] (v3) at (3,2) {$ $};
\node[draw] (v4) at (2,1) {$ $};
\node[draw] (v5) at (3,1) {$ $};
\node[draw] (v6) at (2,0) {$ $};
\node[draw] (v7) at (3,0) {$ $};
\node[draw] (t) at (5,1.5) {$t$};

\draw[priority1] (s) to node[above] {1} (v0); 
\draw[priority1] (s) to node[above] {1} (v2); 
\draw[priority1] (s) to node[above] {1} (v4); 
\draw[priority1] (s) to node[above] {1} (v6); 

\draw[priority2] (v1) to node[right] {0} (v2); 
\draw[priority2] (v3) to node[right] {0} (v4); 
\draw[priority2] (v5) to node[right] {0} (v6); 

\draw[priority1] (v0) to node[above] {0} (v1);
\draw[priority1] (v2) to node[above] {0} (v3); 
\draw[priority1] (v4) to node[above] {0} (v5); 
\draw[priority1] (v6) to node[above] {0} (v7); 

\draw[priority1] (v1) to node[above] {1} (t); 
\draw[priority1] (v3) to node[above] {1} (t); 
\draw[priority1] (v5) to node[above] {1} (t); 
\draw[priority1] (v7) to node[above] {0} (t); 
\end{tikzpicture}
\caption{The 4-Braess graph for a game with four players, where the PoS is $\frac{10}{7}$. Red wavy edges have higher priority.}\label{fig:pos}
\end{figure}

In a game with $k=b$ players on this $b$-Braess graph, there is a unique equilibrium, except for the strategy of the last player. In every PNE, Player 1 to $k-1$ will use the zigzag path of length 1. Assume any subsequent Player $i$ does not choose this path. Then player $i+1$ can directly overtake Player $i$. The last player is not in danger of being overtaken and thus, this player can choose any strategy which uses the third edge of the $b$-th path. In total, arrival times add up to $\sum_{i \in N} i = \frac{k(k+1)}{2}$.

Since at most one player can arrive at time 1, it is obviously socially optimal to use the $b=k$ parallel paths, where one player arrives at time 1 and all other players arrive at time 2. Thus, costs of the unstable social optimum are $2k-1$. 

Therefore, the price of stability in this game is given by
$$
\operatorname{PoS}=\frac{\frac{k(k+1)}{2}}{2k-1}\geq \frac{\frac{k(k+1)}{2}}{2k}=\frac{k+1}{4}.
$$
\end{proof}

\section{Computing Equilibria and the Price of Mistrust}
\label{sec:alg}
\subsection{Computing an equilibrium}

In this section, we present an efficient algorithm for computing PNEs in symmetric routing games with edge priorities. We start with an outline of the main ideas, before we give a detailed description of the algorithm in pseudo-code and proofs for correctness and run-time.

The algorithm itself mainly consists of three steps. The first step initializes a kind of shortest path network and this step is executed only once. In step two, a path for the next player is found within this shortest path network, and in step three, the network is updated to renew the earliest arrival property for the upcoming player. Step two and three are executed once for each player. In detail:

 \begin{enumerate}
  \item A modified Dijkstra search~\cite{Dijkstra59} is executed starting in $s$. We determine two functions $d : V \to \N_0$ and $\varepsilon : E \to \N_0$. Here, $d(v)$ describes, at which time step node $v$ can be reached at the earliest, i.e., at this initialization step $d(v)$ is the standard label set by Dijkstra's algorithm. The function value $\varepsilon(e)$ defines the earliest point in time at which edge $e$ can be left. Hence, $d(v)\le \varepsilon(e)$ where $e=(u,v)$. In the initialization, it holds $\varepsilon(e) = d(u)+\tau(e)$ for $e=(u,v)$, since there are no waiting times. Moreover, the Dijkstra search is not stopped, when the final label of node $t$ is found. Instead, we explore the whole network to prepare the data structure for even longer paths of subsequent players. Please note that the sub network $G'=(V,E')$ where $E'=\{e\in E: e=(u,v) \text{ and }\varepsilon(e)=d(v)\}$ contains all shortest paths from $s$ to $t$. More precisely, every path from $s$ to $t$ in this sub network is a shortest path in the original network. 
  
  \item We now perform a backward search in this sub network $G'=(V,E')$ to find a path $p$ for every Player $i$, starting with Player 1. Additionally, we introduce or reset the function $\Delta:V\to \N_0$. For each node $v\in V$, $\Delta(v)$ is the time step at which we are going to leave node $v$ via an edge $e=(v,w)$. Initially, we set $\Delta(t)=d(t)$ and $\Delta(v)=\operatorname{NaN}$ for $v\not=t$.
  
  Any player cannot reach the sink $t$ before $d(t)$. Since we may enter $t$ from every incoming edge immediately, we choose an arbitrary edge $e=(v,t)$ with $\varepsilon(e)=d(t)$ as the last edge of the path $p$. Consequently, we have to leave $v$ at $d(t)-\tau(e)=:\Delta(v)$. More generally, assume we have chosen an edge $e=(u,v)$ to reach node $v$ at time $\Delta(v)$ at the latest. From all available options we choose the earliest possible time to traverse this edge.
  This yields the time $\Delta(u)$ at which we want to leave node $u$. 
  
  Hence, as predecessor of edge $e$, we only have to consider edges $e'=(w,u)$ with $\varepsilon(e')\le \Delta(u)$. Among these edges, we choose the one with the lowest index in $\pi(e)$, i.e., the edge $e'$ with the highest priority to enter $e$. We add $e'=(w,u)$ to the path. Again, we use it at the earliest possible time and we update $\Delta(w)$ accordingly. 
  
  Due to the definition of $d$ and $\varepsilon$, we can always find a feasible edge. Thus, we eventually reach the source $s$ by iteratively adding preceding edges. 
  
  \item After assigning Player $i$ to the path $p$ constructed in step two, we have to update the values of $d$ and $\varepsilon$. If player $i$ exhausted the capacity of an edge $e=(u,v)$, we increment $\varepsilon(e)$ by 1. We now perform a modified Dijkstra search to check, whether we have to increment other labels $d(v)$ and $\varepsilon(e)$, too. In detail, the new values are $\varepsilon(e):=\max\{\varepsilon(e),d(u)+\tau(e)\}$ and $d(v):=\min\{\varepsilon(e): e=(u,v)\}$ starting with $d(s)=0$. Finally, we reset $\Delta$.  
  
  Now, we can go back to step two to compute the path for the next Player $i+1$.
 \end{enumerate} 
 
In Algorithm~\ref{algo:pathfinder} the pseudo-code of the algorithm can be found. 
 
\begin{algorithm2e}[ht!]
\KwIn{$G=(V,E)$ with priorities $\pi$, $s,t \in V$ and set of players $N = \{1, \dots, k\}$ }
\KwOut{A walk for every player $i \in N$}
calculate $d(v)$ and $\varepsilon(e)$ $\forall v \in V, \forall e \in E$\;\label{line:d_and_eps}
\For{$i$ := 1 to $k$}{
  $e$ := arbitrary edge of $\delta_{\operatorname{in}}(t)$ with $\varepsilon(e) = d(t)$\;
  $v$ := tail($e$)\;
  $\Delta(v)$ := $d(t) - \tau(e)$\;
  $p$ := $\{e\}$\;
  \While{$v \neq s$}{\label{line:path_while_begin}
      let $e' \in \delta_{in}(v)$, with $\varepsilon(e') \leq \Delta$ and maximal priority for $e$\;\label{line:*}
      $v$ := tail($e'$)\;
      $\Delta(v)$ := $\varepsilon(e') - \tau(e')$\;
      $p$ := $p$ $\cup$ $\{e'\}$\;
      $e$ := $e'$\;
  }\label{line:path_while_end}
  print ``Path of player $i$ is $p$''\;
  \ForEach{$e \in p$}{\label{line:capb}
    \If{capacity of $e$ at entry time $\Delta(\operatorname{tail}(e))$ is exhausted}{
    $\varepsilon(e)$ := $\varepsilon(e)$ + 1\;\label{line:cape}
    }
  }
  $d(s)$ := 0\;
  $d(v)$ := $\infty$ $\forall v \in V, v \neq s$\;
  $\Pi$ := $\operatorname{heap}(V,d)$\; 
  \While{$\Pi \neq \varnothing$}{\label{line:update_while_begin}
    $v$ := getMin($\Pi$)\;
    remove $v$ from $\Pi$\;
    \ForEach{$e \in \delta_{\operatorname{out}}(v)$}{
        $\varepsilon(e)$ := $\max\{\varepsilon(e), d(v) + \tau(e)\}$\;
        $d(\operatorname{head}(e))$ := $\min\{d(\operatorname{head}(e)), \varepsilon(e)\}$\;
    }
  }\label{line:update_while_end}
}
\caption{\textsc{Pathfinder}}\label{algo:pathfinder}
\end{algorithm2e}

Before we show the correctness of this algorithm, let us note some important observations. Firstly, for the first player, the situation is quite simple. Here, it holds $\Delta(v)=d(v)$ in all nodes of the constructed path in step two, since the network is not yet congested and no waiting times occur. Player 1 is actually using a shortest path. However, this is not true for subsequent players. Here, edge $e=(v,w)$ could be blocked at time $d(v)$ by a preceding player. 

Secondly, $d(v)$ denotes the earliest arrival time at each node. Furthermore, $\varepsilon(e)$ denotes the first moment in time after $d(v)$ where edge $e(v,w)$ has left over capacity. In particular, $d(t)$ is the earliest time to reach the sink and we always construct a path achieving this time bound.  Note that $d(t)$ may also increment. Yet, for intermediate nodes, it is not always advisable to arrive there at the earliest possible time. Doing so could mean using a low prioritized edge and being trapped in a waiting queue. 

Thirdly, we always use the earliest option for each edge under consideration. Therefore, no subsequent player can use an earlier option and nobody can displace Player $i$.  

And fourthly, by choosing the edge $e'=(u,v)$ with the highest priority in $\pi(e)$, we can guarantee that we can go from $e'$ to $e$ at time $\Delta(v)$. No subsequent player can use edge $e'$ at an earlier time, i.e., there is no additional delay on $e'$ to be expected. And we have the way of right  over all subsequent players who arrive at $v$ via other edges.   

With these remarks in mind, we can now conclude the correctness of Algorithm~\ref{algo:pathfinder}.

\begin{theorem}\label{theo:algo:pathfinder}
Let $(G, u,  \tau, N, \pi)$ be a symmetric routing game with edge priorities and $|N| \geq 1$. Then Algorithm \ref{algo:pathfinder} has running time $\O(|N| (m^2 + n\log n))$ and computes a PNE.
\end{theorem}

\begin{proof}
At first we show, that the algorithm always terminates. Consider the \texttt{while}-loop between lines \ref{line:path_while_begin} and \ref{line:path_while_end}. For any processed edge $e=(u,v)$, it holds $\Delta(u)=\varepsilon(e) - \tau(e)\le \Delta(v)- \tau(e)$. In other words, the values of $\Delta$ are non-increasing inside the loop. Assume there is an edge $e=(u,v)$ to be insert twice in $p$. In this case, $\Delta(v)$ has still to be  $\geq \varepsilon(e)$, but $\Delta(u)$ was already set to $\varepsilon(e) - \tau(e)$ in the first occurrence. In other words, this is only possible when there is cycle of length zero. However, we excluded networks with zero costs cycles from our considerations (see Section \ref{sec:zerocycles}). In consequence, the \texttt{while}-loop processes each edge at most once.

Furthermore, due to the definition of $d(v)$ and $\varepsilon(e)$, there is always at least one feasible edge at each node. Note that $\min\{\varepsilon(e):e\in\delta_{\operatorname{in}}(v) \}\le d(v)\le \min\{\varepsilon(e)-\tau(e):e\in\delta_{\operatorname{out}}(v) \}$ for all nodes $v\in V$. Since $d(s)=0$, $d(v)$ non-increasing, and we have no zero costs cycles, the \texttt{while}-loop eventually reaches $s$ and stops after at most $m$ iterations with an $s$-$t$-path.

Due to lines \ref{line:capb} to  \ref{line:cape}, which prevent using exhausted edges anymore, this path is also feasible.

It remains to show that the algorithm computes a PNE. First of all, we already observed that $d(v)$ denotes the earliest arrival time at a node. The update in lines \ref{line:update_while_begin} and \ref{line:update_while_end} is a slightly modified Dijkstra's algorithm, which also respects exhausted capacities at certain times. Hence, in the iteration for Player $i$, it is impossible to reach $t$ earlier than $d(t)$ on a free path. The constructed path for Player $i$, on the other hand, achieves this arrival time. 

We now use an inductive argument to show that neither Player $i$ can push any Player $j$ with $j<i$ aside to reach $t$ earlier nor any other Player $j$ with $j>i$ can delay Player $i$. We only have to consider direct delay here, which means Player $j$ uses an edge $e$ at exactly the time Player $i$ intended to use it. Indirect delay, i.e., Player $j$ causes a domino effect which delays Player $i$ in the end, always starts with a direct delay of another player $i'$. Thus, assume any Player $j$ could directly delay any Player $i$ with $i<j$ on edge $e=(v,w)$. We have to distinguish two cases. Firstly, both players entered $e$ coming from the same edge $e'$, but Player $j$ used this edge at an earlier point in time. However, this contradicts the construction of the path of Player $i$, since Player $i$ used edge $e'$ at the earliest possible time. Secondly, Players $i$ and $j$ enter edge $e$ from different edges. However, for Player $i$, we used the edge with the highest priority available at this time. Additionally, edges are consumed while assigning players, but no new edges become available for a fixed point in time. Hence, Player $j$ cannot push Player $i$ aside, since the edge of Player $j$ has lower priority.

Summarizing, during the algorithm each player plays the best response to the strategies of the preceding players and no player is delayed after making its choice. Moreover, for the same reasons, the final state is stable, since no player has an improving path.

Finally, let us analyze the run-time. To calculate $d(v)$ and $\varepsilon(e)$ in line \ref{line:d_and_eps}, we use Dijkstra's algorithm, which has a run-time in $\O(n\log n + m)$. In the main loop, which is executed for each player, we have at most $m$ iterations. In each iteration, we have to search in $\delta_{\operatorname{in}}$ for the edge with highest priority. Hence, we have at most $\O(|N|m^2)$ operations, here. The update of $d$ and $\varepsilon$ between line \ref{line:update_while_begin} and \ref{line:update_while_end} can be implemented like Dijkstra's algorithm. In total, this yields run-time of $\O(|N| (m^2 + n \log n+m) + n\log n)=\O(|N| (m^2 + n \log n))$
\end{proof}

Here, the run-time bound is given with respect to the network size. However, it is essentially determined by looping through the edge priority lists. Even when we visit a node twice or more as in Figure~\ref{fig:kreis}, priority is locally increasing, i.e., we only have to run once through all priority lists. Since the total length of all edge priority lists is in $\O(m)$ for global games we can improve the run-time analysis for this case to $\O(|N| (m + n \log n))$. However, these lists are part of the input, too. Hence, run-time is even linear regarding this part of the input. Since we have to return a path for every player, the run-time is obviously only pseudo-polynomial in the number of players $|N|$, but it is polynomial, when we use an unary encoding of players.

For proving the termination of the algorithm, we have shown that for every player each edge is processed at most once. So in addition  to Theorem~\ref{prob:multinodes}, there is always a PNE where each player uses each edge at most once.

\begin{corollary}\label{cor:double_edges}
In every symmetric routing game with edge priorities, there exists a PNE such that every edge is used at most once by each player. 
\end{corollary}

Yet, this does not imply that an optimal strategy does not contain an edge multiple times. For example, in the routing game in the proof of Theorem~\ref{prob:multinodes}, Player $k$ can use $P_{k}=(e_{k+1},e_{k-1},\dots,e_{k-1},e_1)$, which yields arrival time $k+2$, when edge $e_{k-1}$ is used at most $k$ times. 

Furthermore, please note that the presented algorithm does not necessarily compute the best PNE. In particular, the costs of the computed PNE significantly depends on the way in which we choose the incoming edge to the sink $t$. Since we do not have any priority here, one may choose this edge randomly or in any fixed order among all those egdes $e=(v,t)$ with $d(t)=\varepsilon(e)$. As an example, consider again the graph in Figure~\ref{fig:pnenotunique}. If we choose the bottom edge to $t$ first, the algorithm computes the zigzag-path for the first player and every following player. Thus, the algorithm computes the worst possible PNE in this scenario. Contrary, if the algorithm treats the incoming edges of $t$ in a fixed order from top to bottom, it computes the best PNE. 

\subsection{Price of Mistrust}

Algorithm~\ref{algo:pathfinder} calculates PNEs that have a special property. The players choose their strategies one by one in order of the tie breaking rule. A player always chooses a strategy minimizing its cost, but among all those strategies it takes the one, where it cannot be delayed by any of the following players (since it does not trust them). We call such a strategy a \emph{mistrustful best response}.  Note that a player distrusts only the players choosing a strategy after it. The players before took already a decision, furthermore they are always able to delay it by coping the strategy. In the proof of Theorem~\ref{theo:algo:pathfinder} we show that for any player there is always a mistrustful best response. We call a PNE \emph{mistrustful}, if every player plays a mistrustful best response. 

\begin{observation}\label{obs:pom_algorithm}
Every PNE that is computed by Algorithm~\ref{algo:pathfinder} is mistrustful.
\end{observation}

An interesting question is how much this mistrust costs. For this we define the \emph{Price of Mistrust (PoM)} as the quotient of the mistrustful PNE with the lowest costs and the social-optimum. For the best PNE it is often necessary that players co-operate and trust each other. Let us assume the graph in Figure~\ref{fig:pnenotunique}. If the first player chooses the lowest horizontal path, it knows that every following player finds a path that has optimal costs and does not block it. Unfortunately it is possible that any of these players uses a blocking path anyway, maybe because of ignorance or because of malignity.

In the following we present some games which show that the price of mistrust might be equal to the price of anarchy and further reach a value of $\tfrac{k+1}{2}$, but that it can also be strictly smaller than the PoA but strictly larger than the PoS. For this purpose, we use the idea of the $b$-Braess graph of the proof of Lemma~\ref{prob:notunique} and extend it. Roughly speaking we add a second zigzag path, which lies behind the first one and connects the parallel paths in the opposite directions. We call this graphs \emph{$b$-double-Braess graphs}. In Figure~\ref{fig:pom} we show two $4$-double-Braess graphs. In these graphs each outgoing edge of $s$ has costs one, each other edge has costs zero and all edges have unit capacities. For $k = b$ players the costs of the social optimum are $k$, since the players can use parallel paths with costs one. Since there is no cheaper path, this is also a PNE. At first we assume that both zigzag paths have higher priority than the other paths. For the 4-double-Braess graph this case is depicted on the left side of Figure~\ref{fig:pom}. No matter, which edge to $t$ the first player chooses, it always has to use the front zigzag path. Otherwise, it would be possible that a following player blocks it. By using this path it blocks each following player. For this players the same argumentation holds. So the players follow each other. The total costs are $\tfrac{k(k+1)}{2}$, so the PoM is equal $\tfrac{k+1}{2}$, which is also the PoA. Since the best PNE in this game has costs $k$, we can conclude the following corollary from Observation~\ref{obs:pom_algorithm}.

\begin{corollary}
There are symmetric routing games with edge priorities, where Algorithm~\ref{algo:pathfinder} is not able to compute the PNE with minimal costs.
\end{corollary}

\begin{figure}
\centering
\begin{subfigure}{0.48\textwidth}
\centering
\begin{tikzpicture}[every node/.style={circle, minimum size=12pt, inner sep=2pt, font = \footnotesize}]
\node[draw](s) at (-0.25,2) {$s$};
\node[draw](t) at (5.25,2){$t$};
\node[draw](v3) at (0.75,3) {};
\node[draw](v4) at (0.75,2){};
\node[draw](v5) at (1.5,2){ };
\node[draw](v6) at (1.5,1){};
\node[draw](v7) at (2.25,1) {};
\node[draw](v8) at (2.25,0) {};

\node[draw](v10) at (4.5,3){};
\node[draw](v12) at (3.75,2){};
\node[draw](v13) at (4.5,2){ };
\node[draw](v14) at (3,1){};
\node[draw](v15) at (3.75,1) {};
\node[draw](v16) at (3,0) {};

\path[priority1] (s) -- (v3);
\path[priority1] (s) -- (v4);
\path[priority1] (s) -- (v6);
\path[priority1, bend right] (s) to (v8);
\path[priority2] (v3) -- (v4);
\path[priority1] (v4) --  (v5);
\path[priority2] (v5) -- (v6);
\path[priority1] (v6) -- (v7);
\path[priority2] (v7) -- (v8);
\path[priority1] (v3) -- (v10);
\path[priority1] (v5) -- (v12);
\path[priority1] (v7) -- (v14);
\path[priority1](v8)--(v16);
\path[priority2] (v13) -- (v10);
\path[priority1] (v12) -- (v13);
\path[priority2] (v15) -- (v12);
\path[priority1] (v14) -- (v15);
\path[priority2] (v16) -- (v14);

\path[priority1, bend right] (v16) to (t);
\path[priority1] (v15) -- (t);
\path[priority1] (v13) -- (t);
\path[priority1] (v10) -- (t);
\end{tikzpicture}
\end{subfigure}
\begin{subfigure}{0.48\textwidth}
\centering
\begin{tikzpicture}[every node/.style={circle, minimum size=12pt, inner sep=2pt, font = \footnotesize}]

\node[draw](s) at (-0.25,2) {$s$};
\node[draw](t) at (5.25,2){$t$};
\node[draw](v3) at (0.75,3) {};
\node[draw](v4) at (0.75,2){};
\node[draw](v5) at (1.5,2){ };
\node[draw](v6) at (1.5,1){};
\node[draw](v7) at (2.25,1) {};
\node[draw](v8) at (2.25,0) {};

\node[draw](v10) at (4.5,3){};
\node[draw](v12) at (3.75,2){};
\node[draw](v13) at (4.5,2){ };
\node[draw](v14) at (3,1){};
\node[draw](v15) at (3.75,1) {};
\node[draw](v16) at (3,0) {};

\path[priority2] (s) -- (v3);
\path[priority2] (s) -- (v4);
\path[priority2] (s) -- (v6);
\path[priority1, bend right] (s) to (v8);
\path[priority1] (v3) -- (v4);
\path[priority1] (v4) --  (v5);
\path[priority1] (v5) -- (v6);
\path[priority1] (v6) -- (v7);
\path[priority2] (v7) -- (v8);
\path[priority1] (v3) -- (v10);
\path[priority1] (v5) -- (v12);
\path[priority1] (v7) -- (v14);
\path[priority1](v8)--(v16);
\path[priority2] (v13) -- (v10);
\path[priority1] (v12) -- (v13);
\path[priority2] (v15) -- (v12);
\path[priority1] (v14) -- (v15);
\path[priority2] (v16) -- (v14);

\path[priority1, bend right] (v16) to (t);
\path[priority1] (v15) -- (t);
\path[priority1] (v13) -- (t);
\path[priority1] (v10) -- (t);
\end{tikzpicture}
\end{subfigure}

\caption{Two examples of the 4-double-Braess graph. In both graphs the outgoing edges of $s$ have all costs one, each other edge has costs zero. The red wavy edges have the higher priority. On the left hand side the PoM is equal to the PoA. On the right hand side the PoM lies between the PoS and the PoA.}\label{fig:pom}
\end{figure}
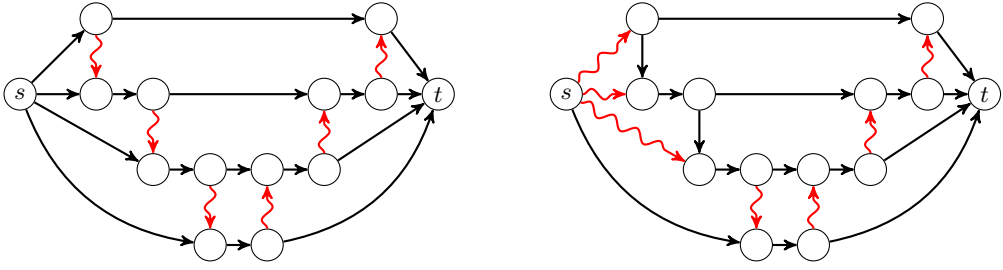

\begin{observation}
There are symmetric routing games with edge priorities, where the Price of Mistrust is greater than the Price of Stability and less than the Price of Anarchy.
\end{observation}

\begin{proof}
We use the $b$-double-Braess graph again, but change the priorities. Each outgoing edge of $s$ except the lowest gets the higher priority. The rear zigzag path remains prioritized. The example for the 4-double-Braess graph is depicted one the right side of Figure~\ref{fig:pom}. With this priorities the social optimum and best PNE with cost $k$ is not a mistrustful PNE. The player who chooses the lowest path can be blocked by the player, that chooses the second lowest path. Each other player can be blocked by the player, who chooses the path, which lies directly under its own path.
In the best mistrustful PNE the first player chooses the second lowest edge from $s$. Then it goes one step down and goes directly to $t$. Each following player $i$ chooses the $(i-1)$th lowest horizontal path. The last player follows any of the former players. So $k-1$ players arrive at time step 1 and one player arrives at time step 2, which gives total costs of $k+1$. In the worst PNE, the first player uses the rear zigzag path, after going down from the second lowest to the lowest path. Every other player follows player 1. The overall costs are $\tfrac{k(k+1)}{2}$ again. So with these priorities the PoM lies between the PoS and PoA ($1<\frac{k+1}{k}<\frac{k+1}{2}$, for $k>2$).
\end{proof}

\section{Connection to Earliest Arrival Flows}
\label{sec:eaf}
A socially optimal state of a network with a single source and a single sink and without priority lists is an earliest arrival flow (EAF). Furthermore, if we restrict to unit capacities and follow the same argumentation as in Harks et.\ al.~\cite{HPSV16}, we can show that the socially optimal state of a given symmetric game with edge priorities is an earliest arrival flow, too. Similar to the model with player priorities, there is always a socially optimal solution with edge priorities, where players only wait at the source node in an earliest arrival flow in a network with unit capacities due to the fact that an earliest arrival flow fulfills strong flow conservation. But at the source, the global tie-breaking order of the players applies and there is no priority rules of any entering edges in our model. Thus, both models, our model with edge priorities and the model with player priorities, coincide in this case. 

However, in~\cite{HPSV16} an earliest arrival flow is also always an equilibrium, i.e., the Price of Stability is equal to one. This is not the case in our model with edge priorities as we have seen in Theorem~\ref{theo:pos}. Yet, given a network initially without edge priorities, it may still be possible to find a priority rule on the edges, such that there is an equilibrium that is an earliest arrival flow. In other words, can a network operator always establish a system optimal flow in a network with edge priorities by choosen an appropriate priority list?

Unfortunately, this is not possible in general. We first present a simple example in Fig.~\ref{fig:EAFnoUE_capacities}. Assume, we have four players. Thus, in an EAF the first three players use the lower path by following each other. Hence, the arrival times are 1,~2, and~3. Furthermore, there is capacity left on the arc $(s,v)$ which can be used by Player 4. Consequently, this player starts at time step zero, uses the upper path, and also arrives at time step 3. This is not a PNE, since Player~4 can improve its arrival time by switching to the lower edge after arriving at $v$. Since it arrived before Player~2 at $v$ it can enter the edge before this player and arrives at time step 2. In consequence, Player~2 has to choose a strategy in a PNE which prevents being overtaken. This implies that Player~2 has to use $(s,v)$ at time step zero and waits one time unit at $v$ to enter the lower edge. However, this blocks the capacity for Player~4 at time zero, that is, Player~4 cannot arrive at time step 3 anymore. Summarizing, the arrival times in a PNE are 1, 2, 3, and 4. Note that we did not use any prioritization in this example. Hence, the problem is caused by the FIFO rule and non unit capacities.

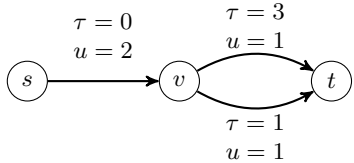
\begin{figure}[ht!]
\centering
\begin{tikzpicture}[every node/.style={circle, minimum size=15pt, inner sep=2pt, font = \footnotesize}]

\node[draw](s) at (0,0) {$s$};
\node[draw](t) at (4,0){$t$};
\node[draw](v) at (2,0) {$v$};

\path[blackedge] (s) -- node[above, align=left]{$\tau = 0$ \\ $u = 2$} (v);
\path[blackedge, bend left] (v) to node[above, rectangle, align=left]{$\tau = 3$ \\ $u = 1$} (t);
\path[blackedge, bend right] (v) to node[below, rectangle, align=left]{$\tau = 1$ \\ $u = 1$} (t);
\end{tikzpicture}
\caption{A small example where any EAF cannot be a PNE. Assume, there are four players. In an EAF they arrive at time steps 1,~2,~3 and 3. In every PNE they arrive at time steps 1, 2, 3 and 4.}
\label{fig:EAFnoUE_capacities}
\end{figure}

Consequently, from now on we only consider networks with unit capacities. Even in such networks, it is not always possible to establish a system optimal earliest arrival flow by appropriate edge priorities. First, note that it is not possible to exchange an edge with high capacity for several parallel edges with unit capacities, since edge priority lists in a network constructed this way create a priority among players using the original edge at the same time, which was not present in the original network. In other words, players using an edge with high capacity at the same all have the same priority and tie breaking by ID applies, whereas players on the multiple copies with unit capacity all have different priorities for entering the subsequent edges. Thus, the example from above with two edges from $s$ to $v$ with unit capacities does not work and we need a more sophisticated construction.

\begin{theorem}\label{thm:noEAF}
 There exist networks with unit capacities where an earliest arrival flow cannot be established by any priority order of the edges as a PNE flow with edge priorities. Further, for arbitrary many players and arbitrary high edge costs the gap of the social costs of an optimal flow (EAF) and the best PNE is arbitrary large.
\end{theorem}

\begin{proof}
Consider the graph depicted in Figure \ref{fig:EAFnoUE} with $k>3M$ players and unit capacities. Edges are labeled with travel times where $M>6$. One can now use a successive shortest path algorithm on the residual network, where backward edges have negative costs, to determine the EAF. This yields three paths, namely  $(s,v_1,v_2,t)$ with length 3, $(s,v_5,v_2,v_1,v_4,t)$ with length 5, and $(s,v_3,v_1,v_2,v_5,t)$ with length $M+1$. Thus, the first player arrives at $t$ at time 3. From time step 5 on, two players reach $t$ per time step. Beginning with time step $M+1$, three players arrive at $t$ simultaneously. Reconstructing the actual paths, i.e., paths without backward edges, yields five different paths which are used by players.

Now, we try to construct edge priority lists which yield the same unique EAF. Since $v_1$ and $v_2$ are the only nodes with two incoming edges besides $t$, priority is only of interest there. It is easy to see that priority at $v_1$ has no influence on the total value of the solution, only the path decomposition is slightly affected. Thus, the crucial edge is $e_1=(v_2,t)$ and we can achieve two different PNEs, depending on the priority list of $e_1$. 
If the edge $(v_1,v_2)$ is in front of $(v_5,v_2)$ in the priority list of $e_1$, the players on the path $(s,v_1,v_2,t)$ can never be displaced. Since it is the shortest path and it is absolutely a save choice, the PNE will consist of players using the paths $P_1=(s,v_1,v_2,t)$, $P_2=(s,v_3,v_1,v_4,t)$, and $P_3=(s,v_5,t)$.
Contrary, if the edge $(v_5,v_2)$ is preferred over $(v_1,v_2)$ in the priority list of $e_1$, players will only use the paths $P_4=(s,v_5,v_2,t)$ and $P_5=(s,v_1,v_4,t)$ in the PNE, i.e., at most two players arrive at $t$ per time step. For $k>3M$, this is always worse than the first equilibrium.

Now, let us compare the costs of the EAF and the best PNE. The arrival times of the EAF were already stated above. In the PNE, the players arrive after time 3 on $P_1$, after time 6 on $P_2$ and after time $M$ on $P_3$. In total, this yields a difference of $C(PNE)-C(EAF)=M-5$, which goes to infinity for $M \to \infty$.
\begin{figure}[ht!]
\centering
\begin{tikzpicture}[every node/.style={circle, minimum size=15pt, inner sep=2pt, font = \footnotesize}]

\node[draw](s) at (0,3) {$s$};
\node[draw](t) at (6,3){$t$};
\node[draw](v1) at (2,3) {$v_1$};
\node[draw](v2) at (4,3){$v_2$};
\node[draw](v3) at (1,1) {$v_3$};
\node[draw](v4) at (4,1){$v_4$};
\node[draw](v5) at (3,5) {$v_5$};

\path[blackedge] (s) -- node[above]{$0$} (v1);
\path[blackedge] (v1) -- node[above]{$3$} (v2);
\path[blackedge] (v2) -- node[above]{$0$} (t);
\path[blackedge] (s) -- node[above]{$0$} (v5);
\path[blackedge] (s) -- node[below]{$0$} (v3);
\path[blackedge] (v3) -- node[below]{$2$} (v1);
\path[blackedge] (v1) -- node[below]{$0$} (v4);
\path[blackedge] (v4) -- node[below]{$4$} (t);
\path[blackedge] (v5) -- node[above]{$4$} (v2);
\path[blackedge] (v5) -- node[above]{$M$} (t);

\end{tikzpicture}
\caption{An EAF which cannot be a PNE.}
\label{fig:EAFnoUE}
\end{figure}
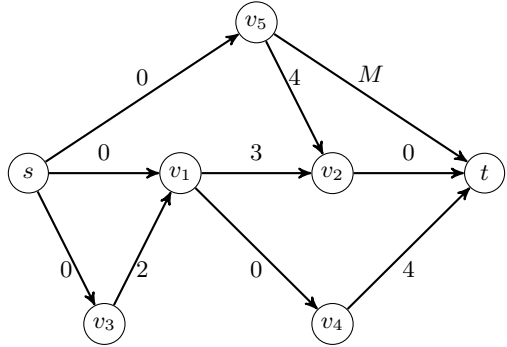
\end{proof}

Yet, we conjecture that it is possible to determine optimal priority rules for some graph classes like series parallel graphs or outerplanar graphs. We present an algorithmic idea how to construct a global priority list of the edges during the computation of an earliest arrival flow. Whenever this construction leads to a feasible global list, we conjecture that there exists an earliest arrival flow that is a mistrustful equilibrium and this equilibrium can be computed by Algorithm~\ref{algo:pathfinder} for the right choice of the edges in $t$.

The main idea is as follows. The algorithm of Wilkinson~\cite{wilkinson1971algorithm} for the construction of a discrete earliest arrival flow uses a \emph{static residual network}. In such a network, the algorithm computes shortest paths until the length of a shortest path exceeds a given time horizon or until the capacity of the network is exhausted. Here, we keep up the shortest path computations until the capacity of the static network is completely used.
%
%
 We construct the global priority list by considering the paths computed by the Wilkinson-Algorithm in order of their occurrence.
For a specific path, we look at an edge $e$ and the subsequent edge $e'$ in this path. Now, we distinguish two cases. If $e'$is a forward edge in the path, we add $e$ at the end of the global priority list if $e$ is used for the first time. Whereas, if the edge $e$ is the last forward edge in front of the backward edge $e'$, we add the edge $e$ in front of $e'$ in the priority list. In all other cases, nothing is done.
 
A list is called a feasible list, if each edge occurs at most once in the global priority list. All edges not appearing in the list are added to the end of the list.
If an edge appears twice in a list, we call the list infeasible. 

\begin{lemma}\label{theo:eaf_series-parallel}
On series-parallel networks with unit capacities, we always compute a feasible list $\pi$.
\end{lemma}

\begin{proof}
Ruzika et al.~\cite{Ruzika11} have shown that one can find an earliest arrival flow in a series-parallel network by sending players on shortest paths as long as there is free capacity. Thus, a backward edge is never used in the computation of the earliest arrival flow. Since this is the only possibility to get an infeasible list in our construction above, the described procedure always works successfully for series-parallel network.
\end{proof}

\begin{conjecture}\label{theo:eaf_outerplanar}
On outerplanar graphs with unit capacities, we always compute a feasible list $\pi$.
\end{conjecture}
 
\begin{proof}[Idea of proof]
First, note that in each step where the algorithm of Wilkinson computes a shortest path in the residual network, there are three possibilities for an edge: 
  \begin{enumerate}
  \item the edge is used for the first time in forward direction,
  \item the edge is used as a backward edge, or 
  \item the edge is used in forward direction, but it was already used before.  
  \end{enumerate}
  
 Obviously, each case can only occur if the preceding case has already occurred for this edge. Moreover, since we consider a network with unit capacities, the usage as a forward edge or as backward edge alternates. Thus, a contradiction, i.e., an edge $e=(u,v)$ inserted for the second time in the priority list, can only occur, when we use this edge at least for the second time in forward direction and the subsequent edge $e_1=(w,v)$ is used as backward edge.

 
 This implies the following connections (edges or paths) in the graph $G$:
 
 \begin{itemize}
 \item Since $e=(u,v)$ was used as a forward edge, there is a connection $s-u$ in the graph.
 \item Since $e=(u,v)$ was used as a backward edge, there is a connection $u-t$ that does not use $e$.
 \item Since $e_1=(w,v)$ is used as a backward edge now, it was used as a forward edge before. Thus, there is a connection $s-w$.
 \item Since $e_1$ is used as a backward edge and $e$ is the last forward edge before, there is a connection $w-t$. 
 \end{itemize}

 \begin{figure}[ht!]
 \centering
 \begin{tikzpicture}[every node/.style={circle, minimum size=15pt, inner sep=2pt, font = \footnotesize}]
 
 \node[draw](s) at (0,2) {$s$};
 \node[draw](t) at (6,2){$t$};
 \node[draw](u) at (2,2) {$u$};
 \node[draw](v) at (4,2){$v$};
 \node[draw](u1) at (4,4) {$w$};

 \draw[draw,thick,-stealth', draw=black, dotted] (s) -- (u);
 \path[blackedge] (u) -- (v);
 \path[blackedge] (u1) -- (v);
 \draw[draw,thick,-stealth', draw=black, dotted] (s) -- (u1);
 \draw[draw,thick,-stealth', draw=black, dotted] (u1) --  (t);
 \draw[blackedge, bend angle=30, bend right, dotted] (u) to  (t);

 \end{tikzpicture}
 \caption{A conflict in $\pi$ leads to a subgraph homeomorphic to $K_{2,3}$}
 \label{fig:K23}
 \end{figure}
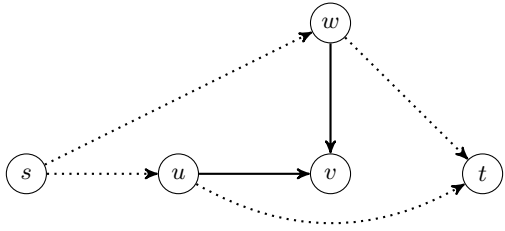
 
 Drawing these connections yields the graph depicted in Figure~\ref{fig:K23}. The underlying undirected graph is a $K_{2,3}$ (consider the sets $\{ u, w \}$ and $\{ s,v,t\}$).
 Sys\l o~\cite{syslo1979characterizations} has shown that a graph with a subgraph homeomorphic to $K_{2,3}$ can never be outerplanar. Thus, we always get a feasible $\pi$ for every outerplanar graph.
 \end{proof}

\begin{conjecture}\label{prop:priolist}
Let $G=(V,E)$ be either a series-parallel network or an outerplanar graph. Furthermore, $G$ has unit capacities. Let $\pi(G)$ be the global priority list that was constructed by the approach above. Then there exists an earliest arrival flow which is a mistrustful equilibrium in the game $\mathcal{G}=(G,N,\tau,u,\pi(G))$ for any number of players $k=|N|$. This equilibrium can be computed by Algorithm~\ref{algo:pathfinder}.
\end{conjecture}

Note that there is an example, which is neither a series-parallel network nor an outerplanar graph, where the construction always leads to a feasible list but there is no priority rule that assures a system optimal PNE. This example is given in Figure~\ref{fig:noEAFbutPriorities}. The given three paths represent the unique run of the successive shortest path algorithm. It is obvious that this run lead to a feasible priority rule, since every edge is used at most twice. The arrival pattern of the EAF is $3,4,5,5,6,6,7,7,7, \ldots$. This EAF cannot be established by any priority rule. Assume, we prioritize edge $(v_3,v_2)$ over $(v_1,v_2)$. Then, Player~2 has to use the path $s, v_3, v_2, t$ to achieve arrival time~4, otherwise she would be blocked by Player~3. Because of this choice, only one player can reach $t$ at time step 5. Thus, this is not an EAF. If we prioritize $(v_1,v_2)$ over $(v_3,v_2)$ the edge $(v_1,t)$ will never be used, since the path $s,v_1,v_2,t$ is always the shorter alternative. Thus, at most two players arrive at $t$ at each time step, which is also not an EAF.

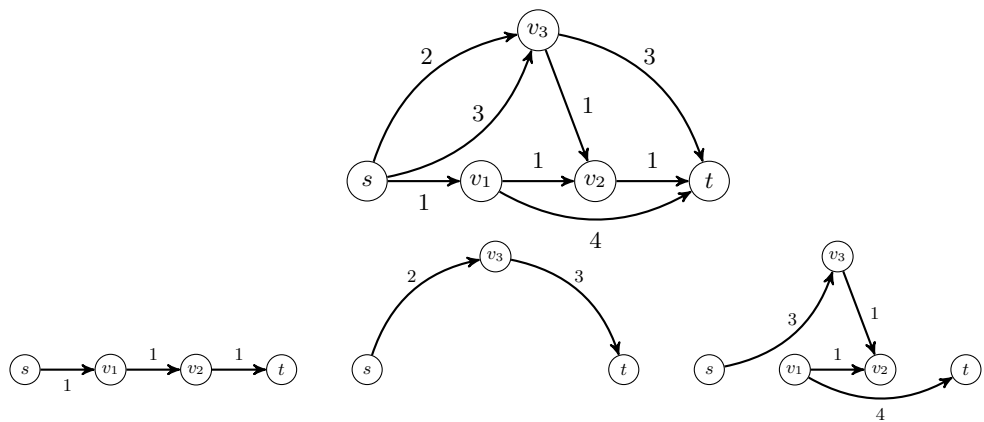
\begin{figure}[ht!]
 \centering
 \begin{tikzpicture}[every node/.style={circle, minimum size=15pt, inner sep=2pt, font = \footnotesize, transform shape}]
 
 \node[draw](s) at (0,0) {$s$};
 \node[draw](t) at (4.5,0){$t$};
 \node[draw](v1) at (1.5,0) {$v_1$};
 \node[draw](v2) at (3,0){$v_2$};
 \node[draw](v3) at (2.25,2) {$v_3$};

 \path[blackedge] (s) -- node[below]{$1$} (v1);
 \path[blackedge] (v1) -- node[above]{$1$} (v2);
 \path[blackedge] (v2) -- node[above]{$1$} (t);
 \path[blackedge, bend left] (v3) to node[above]{$3$} (t);
 \path[blackedge, bend right] (v1) to node[below]{$4$} (t);
 \path[blackedge, bend left] (s) to node[above]{$2$} (v3);
 \path[blackedge, bend right] (s) to node[above]{$3$} (v3);
  \path[blackedge] (v3) -- node[right]{$1$} (v2);
  
\begin{scope}[xshift= -4.5cm, yshift=-2.5cm, scale=0.75]
 \node[draw](s) at (0,0) {$s$};
 \node[draw](t) at (4.5,0){$t$};
 \node[draw](v1) at (1.5,0) {$v_1$};
 \node[draw](v2) at (3,0){$v_2$};

 \path[blackedge] (s) -- node[below]{$1$} (v1);
 \path[blackedge] (v1) -- node[above]{$1$} (v2);
 \path[blackedge] (v2) -- node[above]{$1$} (t);
\end{scope}

\begin{scope}[yshift=-2.5cm, scale=0.75]
 \node[draw](s) at (0,0) {$s$};
 \node[draw](t) at (4.5,0){$t$};
 \node[draw](v3) at (2.25,2) {$v_3$};
 \path[blackedge, bend left] (v3) to node[above]{$3$} (t);
 \path[blackedge, bend left] (s) to node[above]{$2$} (v3);
\end{scope}

\begin{scope}[xshift=4.5cm, yshift=-2.5cm, scale=0.75]
 \node[draw](s) at (0,0) {$s$};
 \node[draw](t) at (4.5,0){$t$};
 \node[draw](v1) at (1.5,0) {$v_1$};
 \node[draw](v2) at (3,0){$v_2$};
 \node[draw](v3) at (2.25,2) {$v_3$};

 \path[blackedge] (v1) -- node[above]{$1$} (v2);
 \path[blackedge, bend right] (v1) to node[below]{$4$} (t);
 \path[blackedge, bend right] (s) to node[above]{$3$} (v3);
 \path[blackedge] (v3) -- node[right]{$1$} (v2);
\end{scope}
\end{tikzpicture}
\caption{Network which has no priority rule that assures a system optimal PNE. The three given paths are the unique run of the successive shortest path algorithm. This run yields a feasible priority list.}\label{fig:noEAFbutPriorities}
\end{figure}

\section{Discussion}
\label{sec:disc}

Motivated by right of way rules, routing over time with edge priorities brings some surprises like optimal Nash equilibrium solutions with many cycles. Yet, there are also many open questions.

In Section~\ref{sec:alg}, we presented an example where the order in which the algorithm chooses the edges in $t$ has a huge impact on the quality of the equilibrium. 
Is there an optimal ordering of the edges entering $t$, i.e., can we find the best mistrustful equilibrium by determining a fixed order of these edges or do we need a varying ordering for different time steps? Further, we have observed a significant gap between the best mistrustful equilibrium and the best equilibrium. Is there an efficient algorithm that always finds the best Nash equilibrium or is this problem NP-hard? 

Since we have seen that there are games with a very high price of stability and price of anarchy, it seems to be a very interesting question to ask for the design of optimal priority lists. Given a graph, how can the priority list (global or local) which minimizes the best or the worst equilibrium be found?

Moreover, the motivation of this model comes from traffic and hence, it seems quite natural to ask for a generalization of the results to multi-commodity games. Yet, it seems unlikely that equilibria always exist in the multi-commodity case. Under what conditions can we guarantee the existence of equilibria? Do best response dynamics always converge for multi-commodity settings?

%

\bibliographystyle{plain}
\bibliography{packagerouting_bib.bib}

\end{document}